\documentclass[letterpaper, 10 pt, conference]{ieeeconf} 
\IEEEoverridecommandlockouts
\usepackage{cite}
\usepackage[cmex10]{amsmath}
\usepackage{array}
\usepackage{moreverb}
\usepackage{algorithm,algorithmic}
\usepackage{arrayjobx}
\makeatletter
\let\NAT@parse\undefined
\makeatother
\usepackage[colorlinks,citecolor=darkblue,urlcolor=red, linkcolor=black, hyperfigures]{hyperref}
\usepackage{amsmath,amssymb}
\usepackage{times}
\usepackage{graphicx}
\usepackage{subfigure}
\usepackage{setspace}
\usepackage{soul, xcolor}
\usepackage{float}
\usepackage{indentfirst}
\usepackage{bm}
\usepackage{booktabs}
\usepackage{flushend}
\usepackage{balance}
\usepackage[export]{adjustbox}
\usepackage{caption}
\usepackage{enumerate}
 
\usepackage{algorithmic}

\usepackage{booktabs}
\usepackage{tabularx}
\usepackage{array}
 
\makeatletter
\newcommand{\removelatexerror}{\let\@latex@error\@gobble}
\makeatother
\newtheorem{example}{Example}
\newtheorem{theorem}{Theorem}
\newtheorem{lemma}{Lemma}

\newtheorem{corollary}{Corollary}
\newtheorem{definition}{Definition}
\newtheorem{proposition}{Proposition}
\newtheorem{problem}{Problem}
\newtheorem{remark}{Remark}
\newtheorem{claim}{Claim}
\newtheorem{assumption}{Assumption}
\newcommand{\bdefinition}{\begin{definition} \begin{rm} }
\newcommand{\edefinition}{ \end{rm} \hfill \rule{1.5mm}{1.5mm}
\end{definition} }
\newcommand{\bremark}{\begin{remark} \begin{rm} }
\newcommand{\eremark}{ \end{rm} \hfill \rule{1.5mm}{1.5mm}
\end{remark} }
\newcommand{\btheorem}{\begin{theorem}  \begin{rm} }
\newcommand{\etheorem}{ \end{rm} \hfill \rule{1.5mm}{1.5mm}
\end{theorem} }
\newcommand{\blemma}{\begin{lemma} \begin{rm} }
\newcommand{\elemma}{ \end{rm} \hfill \rule{1.5mm}{1.5mm}
\end{lemma} }
\newcommand{\bcorollary}{\begin{corollary} \begin{rm} }
\newcommand{\ecorollary}{ \end{rm} \hfill \rule{1.5mm}{1.5mm}
\end{corollary} }
\newcommand{\bproposition}{\begin{proposition} \begin{rm} }
\newcommand{\eproposition}{ \end{rm} \hfill \rule{1.5mm}{1.5mm}
\end{proposition} }
\newcommand{\bclaim}{\begin{claim} \begin{rm} }
\newcommand{\eclaim}{ \end{rm} \hfill \rule{1.5mm}{1.5mm}
\end{claim} }
\newcommand{\bproblem}{\begin{problem} \begin{rm} }
\newcommand{\eproblem}{\end{rm} \end{problem} }
\newcommand{\bassumption}{\begin{assumption} \begin{rm} }
\newcommand{\eassumption}{ \end{rm} \hfill \rule{1.5mm}{1.5mm}
\end{assumption} }

\definecolor{darkblue}{rgb}{0.0, 0.0, 0.55}

\usepackage[textwidth=1.9cm,color=green!10,textsize=footnotesize]{todonotes}

\def\BibTeX{{\rm B\kern-.05em{\sc i\kern-.025em b}\kern-.08em
    T\kern-.1667em\lower.7ex\hbox{E}\kern-.125emX}}

\title{\Huge \bf
Control Synthesis against LTL Specifications with Long-Run Visit Proportion Objectives
}

\author{Zhiyuan Huang, Zhao Tong, Jiakai Li, Chenrui Xiang, and Bingzhuo Zhong\textsuperscript{*}
\thanks{This work was supported by the National Natural Science Foundation of China (Grant No. 62606433), the Guangdong Basic and Applied Basic Research Foundation (No:2026A1515010222), and the Guangdong Provincial Project (No. 2024QN11X053). (Corresponding Author: Bingzhuo Zhong)}
\thanks{
Zhiyuan Huang, Zhao Tong, Jiakai Li, Chenrui Xiang and Bingzhuo Zhong are with the Thrust of Artificial Intelligence, The Hong Kong University of Science and Technology (Guangzhou), Guangzhou 511400, China. (e-mail: \{zhuang655, ztong837, jli196, cxiang359\}@connect.hkust-gz.edu.cn, bingzhuoz@hkust-gz.edu.cn). 
}
}

\begin{document}

\maketitle

\begin{abstract}
This paper investigates the path-planning problem for systems required to satisfy a linear temporal logic (LTL) specification while achieving a desired long-run visit proportion. For a path represented in prefix–suffix structure, the long-run visit proportion quantifies the asymptotic occurrence proportion of an atomic proposition (AP) sequence of interest in the suffix trace. Such a quantitative requirement generally cannot be expressed by standard LTL specifications. Furthermore, we develop a planning approach that synthesizes an LTL-satisfying path whose long-run visit proportion remains within a prescribed tolerance of a desired value while satisfying an overall cost constraint. By adjusting the desired proportion, the synthesized path can allocate more or less long-run attention to the atomic proposition sequence of interest, thereby improving the flexibility and efficiency of the task execution. Finally, experiments on a quadruped robot demonstrate the practical significance of the proposed long-run visit proportion and the effectiveness of the proposed planning approach.
\end{abstract}

\section{Introduction}
Path planning is a fundamental problem in cyber-physical systems, aiming to synthesize a feasible path that guides a system from an initial state toward a desired goal while satisfying environmental and system constraints. Classical planning problems typically focus on low-level objectives, such as obstacle avoidance and point-to-point navigation. However, as the requirements imposed on cyber-physical systems become increasingly complex and diverse, these low-level objectives alone are often insufficient to describe sophisticated tasks. Thus, high-level task-planning methods are needed to specify complex objectives and synthesize paths or strategies that satisfy the task requirements.

Temporal-logic-based planning has attracted considerable attention because of its expressive power and formal guarantees. In particular, LTL is widely used to specify high-level tasks over infinite sequences of atomic propositions. Classical automata-theoretic approaches translate an LTL specification into an equivalent automaton and search for an accepting system path. Building on this framework, existing studies have investigated optimal \cite{smith2011optimal}, sampling-based \cite{kantaros2018sampling}, abstraction-free \cite{luo2021abstraction}, and Petri-net-based planning methods \cite{lv2023optimal}, as well as extensions to uncertain environments and soft task constraints \cite{guo2018probabilistic,cai2021optimal}. Nevertheless, these studies mainly focus on evaluating the satisfaction of LTL tasks and on optimizing the associated path cost without considering other constraints.


Standard LTL specifies qualitative properties of individual system executions but cannot directly express many quantitative or relational requirements encountered in real-world applications. To specify properties involving relationships among multiple executions, the notion of hyperproperties was introduced and can be formally expressed using hyper-temporal logics \cite{clarkson2014temporal}. For example, HyperLTL has been employed to characterize robustness \cite{wang2020hyperproperties} and opacity-related security properties \cite{yang2020secure,yu2022security}. However, hyper-temporal logics primarily characterize interrelations among multiple executions and do not directly provide quantitative properties within an individual execution.

In recent years, several studies have investigated quantitative visit requirements with LTL specifications. Co-Büchi conditions impose finite upper bounds on visits to a given predicate along infinite executions \cite{murali2023co}, while counting LTL extends this idea to multi-agent settings by constraining how many agents satisfy a proposition simultaneously \cite{qiu2026compressing,sahin2019multirobot}. 
Frequency LTL relaxes the standard until operator by requiring its left-hand-side formula to hold at least at a prescribed proportion of positions before the right-hand-side formula is satisfied \cite{bollig2012frequency}, whereas mean-payoff formulations characterize the long-run average of predefined weights under LTL specifications \cite{bohy2013synthesis}.
Moreover, \cite{chen2026optimal} introduced a long-run efficiency objective, defined as the ratio between accumulated rewards and costs, while ensuring the satisfaction of an LTL task in uncertain environments. 
In stochastic systems, steady-state policy-synthesis methods impose asymptotic state-visit constraints together with LTL specifications \cite{velasquez2024controller,atia2020steady}. Although these studies provide different mechanisms for specifying quantitative long-run behavior under LTL specifications, they do not directly address prescribed long-run visit proportions of specified finite AP sequences. 
These sequences capture task execution patterns with particular AP combinations and orderings. Regulating their proportions allows the planner to adjust the emphasis on these patterns according to operational needs while satisfying the same LTL specification.

To address this gap, we investigate path planning under LTL specifications with an additional quantitative requirement on long-run visit behavior. The main contributions of this paper are summarized as follows.
\begin{enumerate}
\item We introduce the notion of the long-run visit proportion, which quantifies the asymptotic occurrence proportion of a specified atomic-proposition sequence in the suffix trace of a prefix–suffix path. Based on this notion, we formulate a path-planning problem that jointly considers LTL satisfaction, an overall cost constraint, and a prescribed long-run visit proportion.
\item We develop an automata-based synthesis approach that finds a prefix–suffix path satisfying the given LTL specification and overall cost constraint while keeping the deviation from the desired long-run visit proportion within a prescribed tolerance. The correctness and optimality of the proposed approach are further analyzed.
\end{enumerate}
Experiments on a quadruped robot further demonstrate
the effectiveness of the proposed approach in achieving
desired long-run visit proportions while satisfying
LTL specifications and overall cost constraints.

Existing studies on steady-state visit constraints
\cite{atia2020steady, velasquez2024controller}
are conceptually related to our work.
However, existing steady-state formulations typically concern the visit proportions of individual states or state sets, whereas our notion can characterize the occurrence proportion of a specified finite sequence of atomic propositions along a prefix--suffix path. It thereby preserves the ordering information among atomic propositions and enables a desired sequence-level visit pattern to be specified directly.

The remainder of this paper is organized as follows. Section \ref{sec: preliminary} introduces the system model, LTL specifications, and Büchi automata. Section \ref{sec: problem formulation} presents a motivating example, defines the long-run visit proportion, and formulates the problem. Section \ref{sec: solution} develops the proposed solution approach. Section \ref{sec: experiment} demonstrates the application of the proposed notion and evaluates the effectiveness of the approach. Finally, Section \ref{sec: conclude} concludes the paper.

\section{Preliminary} \label{sec: preliminary}
This section reviews the abstractions of the system model and task specification. We use $\mathbb{R}$, $\mathbb{R}_{>0}$, $\mathbb{N}$, and $\mathbb{N}_{>0}$ to denote the set of real numbers, positive real numbers, natural numbers, and positive natural number respectively.
$\mathbb{R}^{n}$ denotes the $n$-dimensional vector space.
In analogy with the Kleene star and $\omega$-words, the superscripts $*$ and $\omega$
denote finite and infinite executions, respectively, while the superscript $*\omega$ denotes their union.
Given a finite sequence $w = w(0)w(1)\cdots w(k)$, $|w|$ denotes the length of $w$, i.e., the number of elements in the sequence.

\subsection{Motion ability abstraction}
\begin{definition} \label{def: WTS}
Consider a system work in a workspace $\Pi \subseteq \mathbb{R}^{n}$, which is partitioned as $N \in \mathbb{N}$ regions $\Pi =\left \{ \pi_{1}, \dots, \pi_{N}\right \} $.
The motion ability of the system is abstracted as a finite weighted transition system (WTS) defined as
\begin{equation} \label{eq: wTS}
    T_{} := (\Pi_{}, \Pi_{0}, U, \to_{}, w, AP, L ),
\end{equation}
where $\Pi_{}$ is a finite set of states representing all regions, 
$\Pi_{0} \subseteq \Pi_{}$ is the set of initial states, $U$ is the set of control inputs,
$\to_{} \subseteq \Pi_{} \times U \times \Pi_{}$ is the transition relation, $w: \Pi \times U \times  \Pi \to \mathbb{R}_{>0}$ is a cost function indicating the moving cost between two regions,
$AP$ is a set of atomic propositions having properties of system interests,
and $L: \Pi \to 2^{AP}$ is the labeling function that denotes the properties at a specific region.
The transition system $T$ is assumed to be deterministic, i.e., for any $\pi \in \Pi$ and $u \in U$, there exists at most one $\pi^{+} \in \Pi$ such that $(\pi,u,\pi^{+}) \in \to$.
\hfill \rule{1.5mm}{1.5mm}
\end{definition}

An infinite path of $T$ is an infinite sequence of states
$\tau=\pi(0)\pi(1)\cdots$ such that, for every
$k\in\mathbb{N}$, there exists $u(k)\in U$ satisfying
$(\pi(k),u(k),\pi(k+1))\in{\to}$.
A finite path of $T$ is a finite sequence of states
$\pi(0)\cdots\pi(n)$, with $n\in\mathbb{N}$,
satisfying the same transition condition for $0\le k<n$.
The sets of all finite and infinite paths of $T$ are
denoted by ${Path}^{*}(T)$ and
${Path}^{\omega}(T)$, respectively.
The trace of an infinite path $\tau$ is the infinite
sequence of atomic proposition sets
$\operatorname{trace}(\tau)
=L(\pi(0))L(\pi(1))\cdots\in(2^{AP})^\omega$.
Next, we adopt a prefix–suffix structure to represent the infinite path of the WTS, which enables tractable computation and optimization.
\begin{definition} \label{def: pre-suf struct}
    An infinite path $\tau =\pi(0)\pi(1) \cdots \in Path^{\omega}(T) $ is said to admit a prefix–suffix structure if there exist $n,l \in \mathbb{N}$, the path can be rewritten as 
    \begin{equation}
        \tau =\pi(0)\pi(1) \cdots\pi(n-1) \big(\pi(n)\cdots\pi(n+l)\big)^{\omega} 
    \end{equation}
    where $\tau_{\mathrm{pre}}:=\pi(0) \cdots\pi(n)$ is the prefix structure from an initial state to a state $\pi(n)$, and $\tau_{\mathrm{suf}}:= \pi(n)\cdots\pi(n+l)$ is the suffix structure, which forms a cycle starting at $\pi(n)$, with repeated states allowed.
    For convenience, we write the prefix--suffix structure as $\tau
=\tau_{\mathrm{pre}}\odot \tau_{\mathrm{suf}}^{\omega}$ where $\odot$ denotes path concatenation with the common boundary state
$\pi(n)$ counted only once. 
    In particular, there exists $u(n+l) \in U$ such that $(\pi(n+l),u(n+l) ,\pi(n)) \in \to$, so that the last state of the suffix transitions back to its first
    state.
    Moreover, the cost of the infinite path satisfying the prefix-suffix structure is defined as:
    \begin{equation}
        J(\tau) = \sum_{k=0}^{n+l}w(\pi(k),u(k), \pi(k+1)),
    \end{equation}
    where $w: \Pi \times U \times  \Pi \to \mathbb{R}_{>0}$ is the cost function of WTS introduced in Definition \ref{def: WTS}.
    \hfill \rule{1.5mm}{1.5mm}
\end{definition}

\subsection{Task Specification} \label{sec: LTL}
The high-level tasks we consider in this paper are described as linear temporal logic (LTL) tasks.
Concretely, an LTL task consists of a set of atomic propositions $AP$ and several Boolean and temporal operators, the syntax of which is defined as follows:
\begin{equation*}
    \varphi  ::= True \mid a \mid \varphi_{1} \wedge \varphi_{2} \mid \neg \varphi \mid X \varphi \mid \varphi_{1} U \varphi_{2},
\end{equation*}
where $a \in AP$ is an atomic proposition, and $X$, $U$ denote ‘‘next” and ‘‘until”, respectively. 
The above syntax can also induce the operator such as $\Diamond$ (‘‘eventually”), $\Box $ (“always”) and $\Rightarrow $ (‘‘implication”), where $\Diamond \varphi = True U \varphi$, $\Box \varphi = \neg \Diamond \neg \varphi$ and $\varphi_{1} \Rightarrow \varphi_{2} = \neg \varphi_{1} U \varphi_{2}$.

LTL formulas are used to evaluate whether an infinite word satisfies some properties or not.
We denote $words(\varphi) = \left \{ \sigma \in (2^{AP})^{\omega }\mid \sigma \models \varphi  \right \} $ the set of all infinite words that satisfy the LTL formula $\varphi$, where $\models \subseteq (2^{AP})^{\omega } \times  \varphi$ is the satisfaction relation.

Given an LTL formula, there always exists a corresponding nondeterministic B{\"u}chi automaton (NBA) over $\Sigma =2^{AP}$ which is defined as follows \cite{baier2008principles}. 

\begin{definition} \label{def: buchi}
An NBA $\mathcal{A}_{\varphi}$ corresponding to $\varphi$ is a quintuple $\mathcal{A}_{\varphi} = (Q,\Sigma, \delta, Q_{0}, F )$, where $Q$ is a finite set of states; $ Q_{0} \subseteq Q$ is the set of initial states; $\Sigma =2^{AP}$ is the alphabet; $\delta: Q \times 2^{AP} \to 2^{Q}$ is a nondeterministic transition function and $F \subseteq Q$ is a set of accepting states.   
\hfill \rule{1.5mm}{1.5mm}
\end{definition}

We next introduce the notions of runs, accepting runs, and the language
accepted by an NBA.

\begin{definition}[Run and B{\"u}chi acceptance]
\label{def: buchi acceptance}
Consider an NBA $\mathcal{A}=(Q,\Sigma,\delta,Q_0,F)$ and an infinite
word $\sigma=a(0)a(1)a(2)\cdots\in\Sigma^\omega$. An infinite state
sequence $R=q(0)q(1)q(2)\cdots\in Q^\omega$ is called a run of
$\mathcal{A}$ over $\sigma$ if $q(0)\in Q_0$ and
$q(k+1)\in\delta(q(k),a(k))$ for every $k\in\mathbb{N}$.

Let $\operatorname{Inf}(R)\subseteq Q$ be the set of states that
occur infinitely often in $R$. The run $R$ is an accepting run if
$\operatorname{Inf}(R)\cap F\neq\emptyset$. The word $\sigma$ is
accepted by $\mathcal{A}$ if there exists at least one accepting run of
$\mathcal{A}$ over $\sigma$. The language accepted by $\mathcal{A}$ is
denoted by $\mathcal{L}^{\omega}(\mathcal{A})\subseteq\Sigma^\omega$ and
is defined as the set of all infinite words accepted by $\mathcal{A}$.
\hfill \rule{1.5mm}{1.5mm}
\end{definition}

The following lemma establishes the correspondence between
LTL formulas and nondeterministic B{\"u}chi automata.

\begin{lemma}[LTL-to-NBA translation]
\label{lem: ltl nba}
For every LTL formula $\varphi$ over $AP$, there exists an NBA
$\mathcal{A}_{\varphi}=(Q,2^{AP},\delta,Q_0,F)$ such that
$\operatorname{Words}(\varphi)
=\mathcal{L}^{\omega}(\mathcal{A}_{\varphi})$. Equivalently, for every
infinite word $\sigma\in(2^{AP})^\omega$, it holds that
$\sigma\models\varphi$ if and only if
$\sigma\in\mathcal{L}^{\omega}(\mathcal{A}_{\varphi})$.
\hfill \rule{1.5mm}{1.5mm}
\end{lemma}


\section{Problem formulation}\label{sec: problem formulation}
\subsection{Motivation}
As illustrated in Section \ref{sec: LTL}, LTL is capable of specifying desired properties over infinite paths, such as safety, reachability, and recurrence. These specifications are inherently qualitative, focusing on whether certain events occur or whether specific patterns are eventually or repeatedly satisfied.

However, LTL lacks the ability to capture preferences over infinite behaviors. In particular, it cannot express requirements on the long-run visit proportion of atomic propositions or specific patterns along an infinite trajectory.
In many practical control synthesis and path planning problems, adjusting the long-run visit proportion allows one to encode preferences between objectives, balance resource usage, or improve operational efficiency. 
This limitation of LTL is further illustrated by the following example.
\begin{example} \label{eg: motivation}
    Consider a robot operating in a workspace as illustrated in Fig. \ref{fig: motivating example}. The transitions of the weighted transition system (WTS) correspond to the connectivity between states in the workspace abstraction, and the cost associated with each transition represents the required energy consumption.
   Suppose the initial state of the WTS is $\Pi_{0} = {q_0}$, and the LTL specification is given by $\varphi =\Box \Diamond \text{gather} ~\wedge ~\Box \Diamond \text{recharge} ~\wedge ~\Box \Diamond \text{upload}$. In plain English, this specification requires the robot to visit regions labeled with “gather,” “recharge,” and “upload” infinitely often.
    Assume that the robot has a total energy capacity of 20 units. Whenever the robot visits a region labeled “recharge”, its energy is fully restored. 
    We conservatively require the total energy consumption per cycle to be no greater than the battery capacity.
One possible infinite path satisfying $\varphi$ is
    $\tau_1=(q_0q_1q_0q_2q_1q_0q_2q_3q_0q_2q_3)^\omega.$
    The path $\tau_1$ also satisfies the per-cycle energy
constraint, with an energy consumption of 16 units
per cycle.
    
Now consider a scenario in which gathering involves a
sustained sensing operation, such as inspecting equipment
within a designated region. To begin such operations with
a larger available energy reserve, it is desirable for
the robot to proceed to the gathering region directly
after recharging, via the transit region $q_0$, without
first performing an upload task.
This behavior corresponds to the state sequence
$q_3q_0q_2$, whose label sequence is
$(\{\text{recharge}\},\emptyset,\{\text{gather}\})$.
Along $\tau_1$, resulting in a visit proportion of $3/11$ per cycle.
To increase this proportion, consider an alternative path
$\tau_2=
(q_0q_2q_1q_0q_1q_0q_2q_3q_0q_2q_3)^\omega,$
which also satisfies $\varphi$ and respects the energy
constraint. Along $\tau_2$, the visit proportion of the sequence $q_3q_0q_2$ increases to $6/11$, so that gathering is more frequently initiated directly after recharging.
However, both paths have identical state visit proportions:
$4/11$, $2/11$, $3/11$, and $2/11$ for
$q_0,q_1,q_2,q_3$, respectively.
Thus, this improvement in the desired task ordering
cannot be distinguished by individual state visit
proportions alone.

\hfill \rule{1.5mm}{1.5mm}
\end{example}
From Example \ref{eg: motivation}, both $\tau_{1}$ and $\tau_{2}$ satisfy the LTL specification. However, in practical scenarios, while the LTL specification remains unchanged, the desired performance objectives may vary, requiring different visit proportions to improve operational efficiency. 
Although standard LTL can express
qualitative ordering requirements, it does not in general
specify prescribed long-run occurrence proportions.
Moreover, individual state or state-set visit proportions
alone do not capture the distinction illustrated above.
This motivates planning with visit proportions of
specified finite AP sequences while satisfying the
LTL task and cost constraints.

\begin{figure}[htbp]
	\centering
    \includegraphics[width=0.45\textwidth]{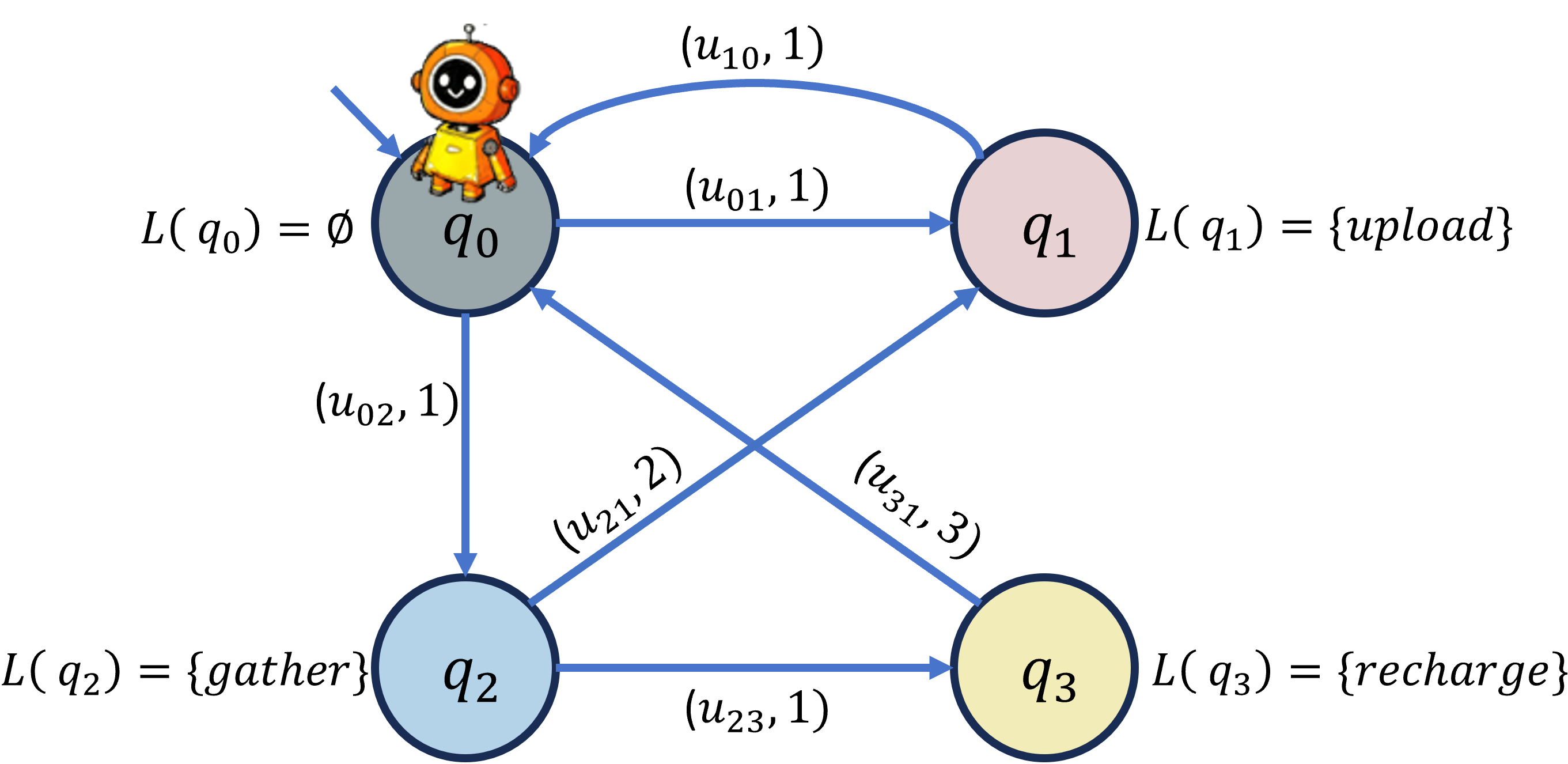}
    \vspace{-4pt}
	\caption{The robot and its workspace in Example \ref{eg: motivation}.}
    \vspace{-5pt}
	\label{fig: motivating example}
\end{figure}

\subsection{Long-run visit proportion and problem formulation}
Motivated by the limitations of LTL illustrated in Example \ref{eg: motivation}, this paper focuses on path planning for LTL tasks under long-run visit proportion requirements and overall cost constraints.
Before formally introducing the notion of long-run visit proportion, we first define the concept of AP sequence of interest and its maximum occurrence number on an infinite path over a prefix-suffix structure.
\begin{definition}\label{def: occ num}
Consider an infinite path
\[
\tau
=
\pi(0)\cdots\pi(n-1)
\big(\pi(n)\cdots\pi(n+l)\big)^\omega
\]
with a prefix--suffix structure as in
Definition~\ref{def: pre-suf struct}.
Its finite suffix trace is
\[
\mathrm{trace}(\tau_{\mathrm{suf}})
:=
L(\pi(n))\cdots L(\pi(n+l)),
\]
whose positions are indexed by $0,\ldots,l$.
Define its infinite periodic extension by
\[
\mathrm{trace}(\tau_{\mathrm{suf}})^\omega(i)
:=
\mathrm{trace}(\tau_{\mathrm{suf}})
\bigl(i\bmod(l+1)\bigr),
 i\in\mathbb N.
\]

Let $v:=a(0)\cdots a(n_v)$ be a finite AP sequence
of interest, where $n_v\in\mathbb N$ and
$a(k)\subseteq AP$ for every $k=0,\ldots,n_v$.
The set of starting positions within one suffix period
at which $v$ occurs cyclically is
\[
\mathcal O
:=
\left\{
i\in\{0,\ldots,l\}
\;\middle|\;
\begin{array}{l}
a(k)\subseteq
\mathrm{trace}(\tau_{\mathrm{suf}})^\omega(i+k),\\
\forall k=0,\ldots,n_v
\end{array}
\right\}.
\]  

The maximum number of non-overlapping cyclic occurrences
of $v$ in $\mathrm{trace}(\tau_{\mathrm{suf}})$ is
\[
\mathrm{num}\bigl(v,\mathrm{trace}(\tau_{\mathrm{suf}})\bigr)
:=
\max_{\mathcal I\subseteq\mathcal O}|\mathcal I|,
\]
subject to
\[
\begin{aligned}
&\{(i+k)\bmod(l+1)\mid k=0,\ldots,n_v\}\cap\\
&
\{(j+k)\bmod(l+1)\mid k=0,\ldots,n_v\}
=\emptyset,
\forall i,j\in\mathcal I,\ i\ne j.
\end{aligned}
\]
The empty set is admissible, so the maximum is zero
when $\mathcal O=\emptyset$.
\hfill \rule{1.5mm}{1.5mm}
\end{definition}
Intuitively, from Definition \ref{def: occ num}, the AP sequence of interest $v$ is a finite sequence of subsets of $AP$ that is desired to occur in the suffix trace.
Moreover, the maximum occurrence number $\mathrm{num}(v,\operatorname{trace}(\tau_{\mathrm{suf}}))$ denotes the maximum number of non-overlapping occurrences of the sequence $v$ in the suffix trace.
Based on Definition \ref{def: occ num}, we next define the long-run visit proportion over an infinite path with a prefix–suffix structure.

\begin{definition}\label{def: long run visit}
    Consider an infinite path
\begin{equation*}
    \tau := \pi(0)\pi(1)\cdots\pi(n-1)\big(\pi(n)\cdots\pi(n+l)\big)^{\omega}
\end{equation*}
which follows a prefix-suffix structure introduced in Definition~\ref{def: pre-suf struct}.
Moreover, the AP sequence of interest $v$ and the maximum occurrence number $\mathrm{num}(v,\operatorname{trace}(\tau_{suf}))$ are defined in Definition \ref{def: occ num}.
Then, the long-run visit proportion with respect to $v \in (2^{AP})^{*}$ is given by 
\begin{equation} \label{eq: visit proportion}
     P_v(\tau) := \frac{\mathrm{num}(v,\operatorname{trace}(\tau_{\mathrm{suf}})) \cdot |v |}{|\operatorname{trace}(\tau_{\mathrm{suf}})|}.
\end{equation}
\hfill \rule{1.5mm}{1.5mm}
\end{definition}
Intuitively, Definition \ref{def: long run visit} introduces a proportion measure associated with the AP sequence of interest $v$ as in Definition \ref{def: occ num}. 
The long-run visit proportion captures the fraction of positions in one period of the suffix trace that are covered by occurrences of $v$.

According to the aforementioned definitions, the problem of this paper is formulated as follows.
\begin{problem} \label{problem: 1}
    Consider a WTS $T$ introduced in Definition \ref{def: WTS}, a high-level LTL task $\varphi$ described in Section \ref{sec: LTL}, an AP sequence of interest $v$ defined in Definition \ref{def: occ num} and a desired long-run visit proportion $P_v^{\mathrm{des}} \in [0,1] $. The goal of this paper is to find an infinite path $\tau \in Path^{\omega}(T)$ in prefix-suffix structure with cost $J(\tau)$ as introduced in Definition \ref{def: pre-suf struct} such that
\begin{enumerate}[(i)]
        \item the trace of the infinite path $\tau$ satisfies the allocated LTL task $\varphi$, i.e. $\operatorname{trace}(\tau) \models \varphi$;
        \item the total cost of the infinite path is constraint by some constant $d \in \mathbb{R}_{>0}$, i.e., $J(\tau) \le d$;
        \item the long-run visit proportion with respect to $v$ remains within a prescribed tolerance $\delta \in \mathbb{R}_{>0}$ of the desired value $P_v^{\mathrm{des}}$; that is, the infinite path $\tau \in Path^{\omega}(T)$ satisfies  $| P_v(\tau) - P_v^{\mathrm{des}} | \le \delta $.
\hfill \rule{1.5mm}{1.5mm}
\end{enumerate}
\end{problem}

\begin{remark}
All three conditions in Problem \ref{problem: 1} are imposed as hard
constraints. Conditions (i) and (ii) enforce satisfaction of the LTL
specification and the prescribed cost bound, respectively, while
condition (iii) requires the long-run visit proportion to remain within
a prescribed tolerance of its desired value.
\hfill \rule{1.5mm}{1.5mm}
\end{remark}

\section{LTL Path Planning with Long-Run Visit Proportion Requirements}
\label{sec: solution}

In this section, we propose an approach to address the Problem \ref{problem: 1}. 

\subsection{Product System and Cost Automaton Construction}
\label{subsec: product cost automaton}
First, in order to find an infinite path in a prefix-suffix structure satisfying the allocated LTL task $\varphi$, we introduce the concept of the product system between WTS and NBA..
\begin{definition} \label{def: product system}
    Consider a WTS $T =(\Pi_{}, \Pi_{0}, U, \to_{}, w, AP, L ) $ as introduced in Definition \ref{def: WTS} and an NBA $\mathcal{A}_{\varphi} = (Q,\Sigma, \delta, Q_{0}, F )$ related to the allocated LTL task $\varphi$ as defined in Definition \ref{def: buchi}. The product system between $T$ and $\mathcal{A}_{\varphi}$ is defined as 
    \begin{equation*}
    T_{\otimes}  = (\Pi_{\otimes}, \Pi_{\otimes,0},\to_{\otimes}, w_{\otimes})
\end{equation*}
where
\begin{itemize}
    \item {$\Pi_{\otimes} = \Pi \times Q$ is the set of product states;}
    \item {$\Pi_{\otimes,0} = \Pi_{0} \times Q_{0}$ is the set of initial states;}
    \item {$\to_{\otimes}$ is the transition relation defined by: for any $\pi_{\otimes}=(\pi,q), \pi_{\otimes}^{\prime}=(\pi^{\prime},q^{\prime}) \in \Pi_{\otimes}$, we have $(\pi_{\otimes}, \pi_{\otimes}') \in  \to_{\otimes}$ if the following conditions hold:
    
    (1) $(\pi, u,\pi^{\prime}) \in \to$, for some $u \in U$;
    
    (2) $q' \in \delta(q,L(\pi))$;}
    \item {$w_{\otimes}: \Pi_{\otimes} \times \Pi_{\otimes} \to \mathbb{R}_{>0}$ is the cost function defined by: for any $\pi_{\otimes}=(\pi,q), \pi_{\otimes}^{\prime}=(\pi^{\prime},q') \in \Pi_{\otimes}$, we have $w_{\otimes}(\pi_{\otimes},\pi_{\otimes}') = w(\pi, \pi^{\prime})$.} 
    \hfill \rule{1.5mm}{1.5mm}
\end{itemize}
    
\end{definition}
Intuitively, the product system synchronizes the transitions of the WTS
$T$ with the state evolution of the NBA $\mathcal{A}_{\varphi}$.
However, the product system does not explicitly keep track of the accumulated cost
along a path. Thus, to incorporate the prescribed cost bound $d$ as required by (ii) in Problem \ref{problem: 1}, we introduce the concept of cost automaton.
\begin{definition}
\label{def: cost automaton}
Consider the product system
$T_{\otimes}
=
(\Pi_{\otimes},\Pi_{\otimes,0},
\to_{\otimes},w_{\otimes})$ derived from a WTS $T$ defined in Definition \ref{def: WTS} and an NBA $\mathcal{A}_{\varphi} $ related to LTL tasks as introduced in Definition \ref{def: buchi},
and a prescribed cost bound $d\in\mathbb{R}_{>0}$.
The cost automaton associated with $T_{\otimes}$ and $d$ is defined as
\begin{equation*}
   \mathcal{A}_{c}
=
(Q_c,Q_{c,0},\to_c), 
\end{equation*}
where
\begin{itemize}
   \item $Q_c\subseteq\Pi_{\otimes}\times[0,d]$ is the set of
reachable cost-augmented states, defined by
{\footnotesize
\begin{align*}
&Q_c :=\nonumber \\
&\left\{
(\pi_{\otimes}(k),c(k))
\;\middle|\;
\begin{array}{l}
k\in\mathbb{N},\
\exists \pi_{\otimes}(0),\ldots,\pi_{\otimes}(k-1), \\\text{ such that:}\
\pi_{\otimes}(0)\in\Pi_{\otimes,0},\\
(\pi_{\otimes}(i),\pi_{\otimes}(i+1))
\in{\to_{\otimes}},\\
i=0,\ldots,k-1,\\
c(k)=\displaystyle\sum_{i=0}^{k-1}
w_{\otimes}(\pi_{\otimes}(i),\pi_{\otimes}(i+1))
\le d
\end{array}
\right\}.
\end{align*}}
Here, $\pi_{\otimes}(k)$ is the terminal product state
of a finite path starting from $\Pi_{\otimes,0}$,
and $c(k)$ is the cumulative transition cost along that path.
For $k=0$, the empty sum is defined as zero.

    \item
   $ Q_{c,0}
    :=
    \left\{
    (\pi_{\otimes,0},0)
    \mid
    \pi_{\otimes,0}\in\Pi_{\otimes,0}
    \right\}$
    is the set of initial cost-augmented states.

    \item $\to_c\subseteq Q_c\times Q_c$ is the transition relation.
    For any $q_{c}:= (\pi_{\otimes},c), q_{c}^{\prime}:=(\pi_{\otimes}',c')$, we have $(q_{c}, q_{c}^{\prime}) \in \to_c$
    if and only if (i) $(\pi_{\otimes}, \pi_{\otimes}') \in \to_{\otimes}$, (ii)
    $c'=c+w_{\otimes}(\pi_{\otimes},\pi_{\otimes}')$, and (iii) $  c'\leq d$.
\hfill \rule{1.5mm}{1.5mm}
\end{itemize}
\end{definition}
Intuitively, the cost automaton augments each state of the product system with the cost accumulated along the corresponding path. 
By retaining only cost-augmented states and transitions whose accumulated cost does not exceed $d$, it incorporates the prescribed cost bound directly into the state-space construction. 
This enables the subsequent
planning procedure to search for paths that satisfy the LTL
specification while simultaneously enforcing condition~(ii) of
Problem~\ref{problem: 1}.
The cost automaton $\mathcal{A}_c$ can be constructed by exploring the
reachable states of the product system while keeping track of the
accumulated transition cost. 
As summarized in Algorithm \ref{alg: cost automaton}, the construction starts from every initial state of the product system with zero accumulated cost. 
Whenever a cost-augmented state $(\pi_{\otimes},c)$ is expanded, each outgoing  transition $(\pi_{\otimes},\pi_{\otimes}') \in \to_{\otimes}$ following the transition rules introduced in Definition \ref{def: cost automaton} generates a successor state
$(\pi_{\otimes}',c')$.
The successor state and the corresponding transition are retained only
if $c'\leq d$. The algorithm additionally maintains a predecessor map
$\mathrm{prev}$, which records one predecessor for each newly discovered
state and is used to reconstruct a path from an initial state.
In Algorithm~\ref{alg: cost automaton},
$\textsc{Dequeue}$ removes and returns the first element
of the queue, while $\textsc{Enqueue}$ appends an element
to the end of the queue.
\begin{algorithm}[t]
\small
 \caption{BuildCostAutomaton}
 \begin{algorithmic}[1] \label{alg: cost automaton}
 \renewcommand{\algorithmicrequire}{\textbf{Input:}}
 \renewcommand{\algorithmicensure}{\textbf{Output:}}
    \REQUIRE Product system $T_{\otimes}=(\Pi_{\otimes},\Pi_{\otimes,0},\to_{\otimes},w_{\otimes})$, cost bound $d \in \mathbb{R}_{>0}$.
    \ENSURE Cost automaton $\mathcal{A}_c = (Q_c, Q_{c,0}, \to_c, \mathrm{prev})$.
    \STATE $Q_{c,0} \gets \{(\pi_{\otimes,0},0) \mid \pi_{\otimes,0} \in \Pi_{\otimes,0}\}$
    \STATE $Q_{c} \gets Q_{c,0}$; \quad $\to_c \gets \emptyset$; \quad $\mathrm{prev} \gets \emptyset$
    \STATE $\mathit{Queue} \gets Q_{c,0}$
    \WHILE{$\mathit{Queue} \neq \emptyset$}
        \STATE $(\pi_{\otimes}, c) \gets \textsc{Dequeue}(\mathit{Queue})$
        \FOR{each $\pi_{\otimes}' \in \Pi_{\otimes}$ such that $(\pi_{\otimes},\pi_{\otimes}')\in \to_{\otimes} $}
            \STATE $c' \gets c + w_{\otimes}(\pi_{\otimes}, \pi_{\otimes}')$
            \IF{$c' \leq d$}
                \IF{$(\pi_{\otimes}', c') \notin Q_c$}
                    \STATE $Q_c \gets Q_c \cup \{(\pi_{\otimes}', c')\}$
                    \STATE $\textsc{Enqueue}(\mathit{Queue}, (\pi_{\otimes}', c'))$
                    \STATE $\mathrm{prev}[(\pi_{\otimes}', c')] \gets (\pi_{\otimes}, c)$
                \ENDIF
                \STATE $\to_c\; \gets\; \to_c \cup \big\{\big((\pi_{\otimes},c), (\pi_{\otimes}',c')\big)\big\}$
            \ENDIF
        \ENDFOR
    \ENDWHILE
    \RETURN $\mathcal{A}_c = (Q_c, Q_{c,0}, \to_c, \mathrm{prev})$
 \end{algorithmic}
\end{algorithm}
\vspace{-2pt}


\begin{algorithm}[t]
\small
 \caption{ReconstructPrefix}
 \begin{algorithmic}[1] \label{alg: reconstruct}
 \renewcommand{\algorithmicrequire}{\textbf{Input:}}
 \renewcommand{\algorithmicensure}{\textbf{Output:}}
    \REQUIRE Cost automaton $\mathcal{A}_c = (Q_c, Q_{c,0}, \to_c, \mathrm{prev})$, a target state $q_c = (\pi_{\otimes}, c) \in Q_c$.
    \ENSURE Prefix path $\tilde{\tau}_{\mathrm{pre}}$ in $T$ from some initial WTS state to $\Pi_g(\pi_{\otimes})$.
    \STATE \textbf{Reconstruct the $\mathcal{A}_c$-path backwards via $\mathrm{prev}$}
    \STATE Initialize $\tilde{\tau}_{c,\mathrm{pre}}$ as an empty sequence.
    \STATE $q_{\mathrm{curr}} \gets q_{c}$
    \WHILE{$q_{\mathrm{curr}} \notin Q_{c,0}$}
        \STATE Prepend $q_{\mathrm{curr}}$ to $\tilde{\tau}_{c,\mathrm{pre}}$.
        \STATE $q_{\mathrm{curr}} \gets \mathrm{prev}[q_{\mathrm{curr}}]$
    \ENDWHILE
    \STATE Prepend $q_{\mathrm{curr}}$ to $\tilde{\tau}_{c,\mathrm{pre}}$.\\  
    \STATE \textbf{Project the $\mathcal{A}_c$-path onto the WTS:} \\
    $\tilde{\tau}_{\mathrm{pre}} \gets \operatorname{Proj}_{T}(\tilde{\tau}_{c,\mathrm{pre}})$
    \RETURN $\tilde{\tau}_{\mathrm{pre}}$
 \end{algorithmic}
\end{algorithm}
\vspace{-5pt}








\subsection{Prefix--suffix path synthesis}
\label{subsec: prefix suffix synthesis}

We next use the cost automaton to construct a feasible infinite path in
prefix-suffix structure as introduced in Definition \ref{def: pre-suf struct}. 
To find a valid suffix structure for satisfying LTL tasks $\varphi$, we define the following sets of states.
\begin{definition} \label{def: state set}
    Consider a product system $T_{\otimes} = (\Pi_{\otimes}, \Pi_{\otimes,0},\to_{\otimes}, w_{\otimes})$ as introduced in Definition \ref{def: product system} and its corresponding cost automaton $\mathcal{A}_{c}
= (Q_c,Q_{c,0},\to_c)$ as stated in Definition \ref{def: cost automaton}. We define the set of accepting states of the product system
\begin{equation}
    \Pi_{\otimes, F} := \{(\pi,q) \in \Pi_{\otimes} \mid q \in F \},      
\end{equation}
and the set of accepting states of the cost automaton
\begin{equation}
Q_{c,F}:=
\{
(\pi_{\otimes},c)\in Q_c
\mid
\pi_{\otimes}\in\Pi_{\otimes,F}
\}.     
\end{equation}
Furthermore, for a finite path
$\tau_c=q_c(0)q_c(1)\cdots q_c(m)$
of $\mathcal{A}_c$, let
$\operatorname{States}(\tau_c)
:=
\{q_c(k)\mid k=0,\ldots,m\}$
denote the set of cost-automaton states visited by $\tau_c$.
\hfill \rule{1.5mm}{1.5mm}
\end{definition}
Intuitively, $\Pi_{\otimes,F}$ collects all states of the product
system $T_{\otimes}$ whose NBA components are accepting states, whereas
$Q_{c,F}$ lifts these accepting product states to the cost automaton by
associating them with their reachable accumulated costs.

To establish the correspondence among paths in the cost automaton, the
product system, and the WTS, we introduce projection mappings that
extract the underlying product-system and WTS components from a
cost-automaton path. These mappings allow a path synthesized in
$\mathcal{A}_c$ to be converted into an executable path in $T$.
\begin{definition}
For a cost-automaton state
$q_c=(\pi_{\otimes},c)=((\pi,q),c)\in Q_c$, define
$\operatorname{proj}_{\otimes}:Q_c\to\Pi_{\otimes}$ and
$\operatorname{proj}_{T}:Q_c\to\Pi$ by
$\operatorname{proj}_{\otimes}(q_c):=\pi_{\otimes}$ and
$\operatorname{proj}_{T}(q_c):=\pi$, respectively.

For a finite path
$\tau_c=q_c(0)q_c(1)\cdots q_c(m)$
in a cost automaton $\mathcal{A}_c$, its projection onto the WTS is defined
componentwise as
\begin{equation}
\operatorname{Proj}_{T}(\tau_c)
:=
\operatorname{proj}_{T}(q_c(0))
\operatorname{proj}_{T}(q_c(1))
\cdots
\operatorname{proj}_{T}(q_c(m)).
\label{eq: WTS projection}
\end{equation}
Similarly, its projection onto the product system is
\begin{equation}
\operatorname{Proj}_{\otimes}(\tau_c)
:=
\operatorname{proj}_{\otimes}(q_c(0))
\operatorname{proj}_{\otimes}(q_c(1))
\cdots
\operatorname{proj}_{\otimes}(q_c(m)).
\end{equation}

For a product-system state
$\pi_{\otimes}=(\pi,q)\in\Pi_{\otimes}$, define the projection
$\operatorname{proj}_{T}^{\otimes}:\Pi_{\otimes}\to\Pi$ by
$\operatorname{proj}_{T}^{\otimes}(\pi_{\otimes}):=\pi$.
This projection is extended componentwise to any finite or infinite
product-system state sequence. In particular, for
$\tau_{\otimes}
=\pi_{\otimes}(0)\pi_{\otimes}(1)\cdots$, define
$\operatorname{Proj}_{T}^{\otimes}(\tau_{\otimes})
:=
\operatorname{proj}_{T}^{\otimes}(\pi_{\otimes}(0))
\operatorname{proj}_{T}^{\otimes}(\pi_{\otimes}(1))
\cdots$.
\hfill \rule{1.5mm}{1.5mm}
\end{definition}

Intuitively, $\operatorname{proj}_{\otimes}$ removes the accumulated-cost
component of a cost-automaton state and retains its underlying product
state, whereas $\operatorname{proj}_{T}$ removes both the accumulated-cost
and NBA components and retains only the corresponding WTS state.
Meanwhile, $\operatorname{proj}_{T}^{\otimes}$ extracts the WTS
component directly from a product system state.

The componentwise extensions of these mappings preserve the ordering
of states along a path. Moreover, by the definitions of the transition
relations of $\mathcal{A}_c$ and the product system, they map every
cost-automaton path to valid paths in the product system and the WTS.
Consequently, for every cost-automaton path $\tau_c$,
$\operatorname{Proj}_{T}(\tau_c)
=\operatorname{Proj}_{T}^{\otimes}
(\operatorname{Proj}_{\otimes}(\tau_c))$.

Then, the following lemma establishes the correspondence between a finite
path in the cost automaton and a suffix cycle in the WTS.


\begin{lemma}
\label{lem: suffix cycle}
Consider a finite path segment $\tilde{\tau}_c:=q_c(n)q_c(n+1)\cdots q_c(m)$ of a cost automaton $\mathcal{A}_c$ introduced in Definition \ref{def: cost automaton} satisfying $q_c(n)=(\pi_{\otimes},c_n)$, $q_c(m)=(\pi_{\otimes},c_m)$ with $c_n < c_m$.
Define $\tilde{\tau}_{c,suf}:=q_c(n)q_c(1)\cdots q_c(m-1)$ as the finite path which removes the last state $q_c(m)$ from $\tilde{\tau}_c$.
Then $\operatorname{Proj}_{T}(\tilde{\tau}_{c,\mathrm{suf}})$ is a suffix
cycle of WTS in the sense of Definition~\ref{def: pre-suf struct}.
\hfill \rule{1.5mm}{1.5mm}
\end{lemma}

\begin{proof}
Let
$q_c(k)=(\pi_{\otimes}(k),c_k)$, $\pi_{\otimes}(k)=(\pi(k),q(k))$.
For every $k=n,\ldots,m-1$, the cost-automaton transition $(q_c(k),q_c(k+1))\in\to_c$
implies, by Definition~\ref{def: cost automaton}, that $(\pi_{\otimes}(k),\pi_{\otimes}(k+1))\in\to_{\otimes}$.
By Definition~\ref{def: product system}, this transition $\to_{\otimes}$ further
implies that there exists an input $u(k)\in U$ such that $(\pi(k),u(k),\pi(k+1))\in\to$.
Therefore, the projection of $\tilde{\tau}_c$ onto the WTS is a valid finite path.

Since $q_c(n)$ and $q_c(m)$ have the same component of the state of the product system, i.e., $\operatorname{proj}_{\otimes}(q_c(n)) = \operatorname{proj}_{\otimes}(q_c(m)) $,
which implies that $\pi(m)=\pi(n)$. Therefore, the transition $(\pi(m-1),u(m-1), \pi(m)) \in  \to$ induced by $\tilde{\tau}_c$ indicates that $(\pi(m-1),u(m-1),\pi(n)) \in  \to$.
It follows that
\[
\operatorname{Proj}_{T}(\tilde{\tau}_{c,\mathrm{suf}})
=
\pi(n)\pi(n+1)\cdots\pi(m-1)
\]
defines a suffix cycle according to Definition
\ref{def: pre-suf struct}.
\end{proof}

Intuitively, although $\tilde{\tau}_c$ does not form a cycle in $\mathcal{A}_c$ because its initial and terminal states have different accumulated costs, these two states correspond to the same state in the product system. 
Consequently, the cost-automaton path between them projects onto a cycle in the WTS. Lemma~\ref{lem: suffix cycle} therefore allows candidate WTS suffix cycles to be identified by searching for finite paths between cost-automaton states that have the same product-state component but different accumulated costs.

However, not every suffix cycle obtained in this manner satisfies the LTL specification $\varphi$. We next define the concept of candidate suffix cycles of WTS, which is the suffix cycle which could satisfy the LTL task by prefix-suffix structure.
\begin{definition} \label{def: candidate cycle}
    Consider a finite cost-automaton path segment $\tilde{\tau}_{c,\mathrm{suf}}$ whose WTS projection is a suffix cycle according to Lemma \ref{lem: suffix cycle}.
    The projected suffix cycle in WTS $\tilde{\tau}_{\mathrm{suf}}:=\operatorname{Proj}_{T}(\tilde{\tau}_{c,\mathrm{suf}})$ is called a \emph{candidate accepting suffix cycle} of WTS if $\operatorname{States}(\tilde{\tau}_{c,\mathrm{suf}})
\cap Q_{c,F}\neq\emptyset.$
\hfill \rule{1.5mm}{1.5mm}
\end{definition}
Intuitively, a candidate accepting suffix cycle, as defined in
Definition \ref{def: candidate cycle}, requires its corresponding path
in the cost automaton to visit at least one accepting state introduced
in Definition \ref{def: state set}. To construct a complete accepting
prefix--suffix path, it remains to recover a finite prefix from an
initial state to the initial product state of the candidate suffix
cycle. Such a prefix can be reconstructed by tracing the predecessor
map $\mathrm{prev}$ backward in the cost automaton and then projecting
the resulting path onto the WTS, as summarized in Algorithm
\ref{alg: reconstruct}. The following lemma establishes that the
resulting prefix--suffix path satisfies the LTL specification $\varphi$.

\begin{lemma}
\label{lem: accepting prefix suffix}
Let $\tilde{\tau}_{\mathrm{suf}}$ be a candidate accepting suffix cycle as introduced in Definition \ref{def: candidate cycle}, and
let $\tilde{\tau}_{\mathrm{pre}}$ be the output of Algorithm \ref{alg: reconstruct}, given by $\tilde{\tau}_{\mathrm{pre}} =\text{ReconstructPrefix}(\mathcal{A}_c,q_c(n))$, where $q_c(n)$ is the initial state of the finite path
$\tilde{\tau}_{c,\mathrm{suf}}$.
Then the trace of the infinite path
\begin{equation*}
    \tau
=
\tilde{\tau}_{\mathrm{pre}}
\odot
\big(\tilde{\tau}_{\mathrm{suf}}\big)^{\omega}
\end{equation*}
admits an accepting run in $\mathcal{A}_{\varphi}$, where $\odot$
denotes path concatenation with the common boundary state counted only
once, as introduced in Definition \ref{def: pre-suf struct}.
Consequently,
$\operatorname{trace}(\tau)\in
\mathcal{L}^{\omega}(\mathcal{A}_{\varphi})$ and hence
$\operatorname{trace}(\tau)\models\varphi$.
\hfill \rule{1.5mm}{1.5mm}
\end{lemma}
\begin{proof}
By Algorithm \ref{alg: reconstruct}, the prefix
$\tilde{\tau}_{\mathrm{pre}}$ is obtained by tracing the predecessor map
$\mathrm{prev}$ backward from $q_c(n)$ to some initial state in
$Q_{c,0}$ and then projecting the reconstructed cost-automaton path
onto the WTS. Denote this reconstructed cost-automaton path by
$\tilde{\tau}_{c,\mathrm{pre}}$, and define
$\tilde{\tau}_{\otimes,\mathrm{pre}}
:=\operatorname{Proj}_{\otimes}(\tilde{\tau}_{c,\mathrm{pre}})$.
Since every transition in the cost automaton induces a transition in
the product system, $\tilde{\tau}_{\otimes,\mathrm{pre}}$ is a valid
product-system path from some initial state
$\pi_{\otimes}(0)\in\Pi_{\otimes,0}$ to
$\pi_{\otimes}(n):=
\operatorname{proj}_{\otimes}(q_c(n))$. Moreover,
$\operatorname{Proj}_{T}^{\otimes}
(\tilde{\tau}_{\otimes,\mathrm{pre}})
=\tilde{\tau}_{\mathrm{pre}}$.

By Definition \ref{def: candidate cycle} and Lemma
\ref{lem: suffix cycle},
$\tilde{\tau}_{\mathrm{suf}}
=\operatorname{Proj}_{T}(\tilde{\tau}_{c,\mathrm{suf}})$
is a valid suffix cycle in the WTS. Define
$\tilde{\tau}_{\otimes,\mathrm{suf}}
:=\operatorname{Proj}_{\otimes}
(\tilde{\tau}_{c,\mathrm{suf}})$.
By the construction of the cost automaton,
$\tilde{\tau}_{\otimes,\mathrm{suf}}$ is a valid product-system cycle
starting at $\pi_{\otimes}(n)$, and
$\operatorname{Proj}_{T}^{\otimes}
(\tilde{\tau}_{\otimes,\mathrm{suf}})
=\tilde{\tau}_{\mathrm{suf}}$.
Therefore,
$R_{\otimes}
:=\tilde{\tau}_{\otimes,\mathrm{pre}}
\odot(\tilde{\tau}_{\otimes,\mathrm{suf}})^\omega$
is a valid infinite product-state sequence whose WTS projection is
$\tau=\tilde{\tau}_{\mathrm{pre}}\odot
(\tilde{\tau}_{\mathrm{suf}})^\omega$.

Write
$R_{\otimes}=\pi_{\otimes}(0)\pi_{\otimes}(1)\cdots$, where
$\pi_{\otimes}(k)=(\pi(k),q(k))$. Since
$\pi_{\otimes}(0)\in\Pi_{\otimes,0}$, we have $q(0)\in Q_0$.
Moreover, by the transition relation of the product system,
$q(k+1)\in\delta(q(k),L(\pi(k)))$ for every $k\in\mathbb{N}$.
Hence, $R=q(0)q(1)q(2)\cdots$ is a run of
$\mathcal{A}_{\varphi}$ over
$\operatorname{trace}(\tau)=L(\pi(0))L(\pi(1))\cdots$.

Since $\tilde{\tau}_{\mathrm{suf}}$ is a candidate accepting suffix
cycle, Definition \ref{def: candidate cycle} gives
$\operatorname{States}(\tilde{\tau}_{c,\mathrm{suf}})
\cap Q_{c,F}\neq\emptyset$.
Thus, there exists a state
$q_c^F\in
\operatorname{States}(\tilde{\tau}_{c,\mathrm{suf}})
\cap Q_{c,F}$.
By the definition of $Q_{c,F}$, its product-state component can be
written as
$\operatorname{proj}_{\otimes}(q_c^F)=(\pi^F,q^F)$,
where $q^F\in F$.
Since
$q_c^F\in\operatorname{States}
(\tilde{\tau}_{c,\mathrm{suf}})$,
its product-state component
$\operatorname{proj}_{\otimes}(q_c^F)=(\pi^F,q^F)$
occurs in the product-system suffix cycle
$\tilde{\tau}_{\otimes,\mathrm{suf}}
=\operatorname{Proj}_{\otimes}
(\tilde{\tau}_{c,\mathrm{suf}})$.
Because this product-system suffix cycle is repeated infinitely often
in $R_{\otimes}$, the product state $(\pi^F,q^F)$ occurs infinitely
often in $R_{\otimes}$. Consequently, its NBA component $q^F$ occurs
infinitely often in $R$.
Therefore, $q^F\in\operatorname{Inf}(R)\cap F$, and hence $R$ is an accepting run of $\mathcal{A}_{\varphi}$ over $\operatorname{trace}(\tau)$. It follows
that
$\operatorname{trace}(\tau)\in
\mathcal{L}^{\omega}(\mathcal{A}_{\varphi})$ and, consequently,
$\operatorname{trace}(\tau)\models\varphi$.
\end{proof}

Intuitively, Lemma \ref{lem: accepting prefix suffix} provides a
constructive procedure for obtaining a prefix--suffix path satisfying
the LTL specification. Specifically, one first searches the cost
automaton for a finite path
$\tilde{\tau}_{c,\mathrm{suf}}$ between two states of the cost automaton that
share the same component of product system state. By Lemma
\ref{lem: suffix cycle}, its projection onto the WTS forms a suffix
cycle. If the corresponding cost-automaton path $\tilde{\tau}_{c,\mathrm{suf}}$ contains a state in
$Q_{c,F}$, the resulting suffix cycle is a candidate accepting suffix
cycle according to Definition \ref{def: candidate cycle}. A compatible
prefix can then be recovered using Algorithm \ref{alg: reconstruct},
and Lemma \ref{lem: accepting prefix suffix} guarantees that the
resulting infinite prefix--suffix path satisfies $\varphi$.

However, different candidate accepting suffix cycles may induce
different cyclic traces and, consequently, different long-run visit
proportions with respect to the AP sequence $v$. It is therefore
necessary to enumerate the candidate suffix cycles represented in the
cost automaton and select the one whose long-run visit proportion is
closest to the desired value $P_v^{\mathrm{des}}$. The complete planning procedure
is summarized in Algorithm \ref{alg: main}.

In Algorithm \ref{alg: main}, $\textsc{LTL2NBA}(\varphi)$ in line 2 denotes any toolbox that can convert the LTL $\varphi$ into its associated NBA as defined in Definition \ref{def: buchi}.
The function $\textsc{Product}(T,\mathcal{A}_{\varphi})$ in line 3
constructs the product system $T_{\otimes}$ according to
Definition~\ref{def: product system}.
$\operatorname{Reach}(\Pi_{\otimes,0})$ in line 4 denotes the set of all product-system states reachable from at least one initial state in $\Pi_{\otimes,0}$. 
This set can be computed using a standard graph
reachability search, such as breadth-first search or depth-first search.
For any two states $q_{c,1},q_{c,2}\in Q_c$,
$\textsc{AllPaths}_{\mathcal{A}_c}(q_{c,1},q_{c,2})$ denotes the set of all
finite paths in $\mathcal{A}_c$ from $q_{c,1}$ to $q_{c,2}$. Since the
accumulated-cost component strictly increases along every transition
and is bounded by $d$, the cost automaton is finite and acyclic.
Therefore, $\textsc{AllPaths}_{\mathcal{A}_c}(q_{c,1},q_{c,2})$ is finite
and can be enumerated using a standard depth-first search.

\begin{algorithm}[t]
\small
 \caption{LTL path planning for a desired long-run visit proportion}
 \begin{algorithmic}[1] \label{alg: main}
 \renewcommand{\algorithmicrequire}{\textbf{Input:}}
 \renewcommand{\algorithmicensure}{\textbf{Output:}}
    \REQUIRE LTL formula $\varphi$, WTS $T$, AP sequence of interest $v$, desired proportion $P_v^{\mathrm{des}} \in [0,1]$, tolerance
$\delta\in\mathbb{R}_{>0}$, cost bound $d \in \mathbb{R}_{>0}$.
    \ENSURE Optimal infinite path $\tau$ in prefix--suffix form, or ``no feasible plan''.

    \STATE \textbf{// Step 1: Build product system}
    \STATE $\mathcal{A}_{\varphi} \gets \textsc{LTL2NBA}(\varphi)$
    \STATE $T_{\otimes} \gets \textsc{Product}(T, \mathcal{A}_{\varphi})$
    \IF{$\mathrm{Reach}(\Pi_{\otimes,0}) \cap \Pi_{\otimes,F} = \emptyset$}
        \RETURN ``no feasible plan''
    \ENDIF

    \STATE \textbf{// Step 2: Construct cost automaton}
    \STATE $\mathcal{A}_c = (Q_c,Q_{c,0},\to_c,\mathrm{prev}) \gets \text{BuildCostAutomaton}(T_{\otimes}, d)$
    \STATE $Q_{c,F} \gets \{(\pi_{\otimes},c) \in Q_c \mid \pi_{\otimes} \in \Pi_{\otimes,F}\}$

    \STATE \textbf{// Step 3: Enumerate closed suffix cycles visiting accepting states}
    \STATE $\mathit{best\_diff} \gets +\infty$; \quad $\mathit{best\_plan} \gets \bot$

    \FOR{each $\pi_{\otimes} \in \{\pi \mid \exists c:(\pi_{\otimes},c)\in Q_c\}$}
        \STATE $C(\pi_{\otimes}) \gets \{c \mid (\pi_{\otimes},c)\in Q_c\}$

        \FOR{each $(c_1,c_2) \in C(\pi_{\otimes}) \times C(\pi_{\otimes})$ with $c_1 < c_2$}
            \STATE $q_{c,1} \gets (\pi_{\otimes},c_1)$; \quad $q_{c,2} \gets (\pi_{\otimes},c_2)$

            \FOR{each $\tilde{\tau}_c \in \textsc{AllPaths}_{\mathcal{A}_c}(q_{c,1},q_{c,2})$}

                 \STATE $\tilde{\tau}_{c,suf} \gets $ Remove the last state $q_{c,2}$ from $\tilde{\tau}_{c}$
                
                \IF{$\operatorname{States}(\tilde{\tau}_{c,suf}) \cap Q_{c,F} = \emptyset$}
                    \STATE \textbf{continue}
                \ENDIF

                \STATE $\tilde{\tau}_{\mathrm{suf}} \gets \operatorname{Proj}_{T}(\tilde{\tau}_{c,suf})$

                \STATE $\tilde{\tau}_{\mathrm{pre}} \gets \text{ReconstructPrefix}(\mathcal{A}_c,q_{c,1})$

                \STATE $P_v \gets$ Compute the long-run visit proportion in \eqref{eq: visit proportion} with respect to $v$ by $ P_v = \frac{\mathrm{num}(v,\operatorname{trace}(\tilde{\tau}_{\mathrm{suf}})) \cdot |v |}{|\operatorname{trace}(\tilde{\tau}_{\mathrm{suf}})|}$

                \IF{$|P_v-P_v^{\mathrm{des}}| < \mathit{best\_diff}$}
                    \STATE $\mathit{best\_diff} \gets |P_v-P_v^{\mathrm{des}}|$
                    \STATE $\mathit{best\_plan} \gets (\tilde{\tau}_{\mathrm{pre}},\tilde{\tau}_{\mathrm{suf}})$
                \ENDIF

            \ENDFOR
        \ENDFOR
    \ENDFOR
   \STATE \textbf{// Step 4: Verify the visit-proportion tolerance}
    \IF{$\mathit{best\_diff}\le\delta$}
    \STATE
    $(\tau_{\mathrm{pre}}^{*},\tau_{\mathrm{suf}}^{*})
    \gets\mathit{best\_plan}$
    \STATE
    $\tau^{*}
    \gets
    \tau_{\mathrm{pre}}^{*}
    \odot
    (\tau_{\mathrm{suf}}^{*})^{\omega}$
    \RETURN $\tau^{*}$
    \ELSE
    \RETURN ``no feasible plan''
    \ENDIF
 \end{algorithmic}
\end{algorithm}

We next present the main theorem, which establishes that the complete
planning procedure in Algorithm \ref{alg: main} solves the problem
considered in this paper.
\begin{theorem} \label{thm: solution correctness}
Consider a WTS $T$, an LTL specification $\varphi$, an AP sequence of
interest $v$, a desired proportion $P_v^{\mathrm{des}}$, a tolerance $\delta$, and
a cost bound $d$, as given in Problem \ref{problem: 1}. Let
$\mathcal{A}_{\varphi}$ be an NBA corresponding to $\varphi$,
$T_{\otimes}$ the product system of $T$ and
$\mathcal{A}_{\varphi}$, and $\mathcal{A}_c$ the cost automaton
associated with $T_{\otimes}$ and $d$, as defined in Definitions
\ref{def: buchi}, \ref{def: product system}, and
\ref{def: cost automaton}, respectively.
If Algorithm \ref{alg: main}
returns an infinite path
$\tau^*:=\tau_{\mathrm{pre}}^*
\odot(\tau_{\mathrm{suf}}^*)^\omega$, then $\tau^*$ solves Problem
\ref{problem: 1}.
\hfill \rule{1.5mm}{1.5mm}
\end{theorem}

\begin{proof}
    Suppose that Algorithm \ref{alg: main} returns
$\tau^*=\tau_{\mathrm{pre}}^*
\odot(\tau_{\mathrm{suf}}^*)^\omega$.
The returned suffix $\tau_{\mathrm{suf}}^*$ is generated from a finite
cost-automaton path $\tilde{\tau}_c^*$ from
$q_{c,1}=(\pi_{\otimes},c_1)$ to
$q_{c,2}=(\pi_{\otimes},c_2)$, where $c_1<c_2$. Let
$\tilde{\tau}_{c,\mathrm{suf}}^*$ be obtained by removing the last
state $q_{c,2}$ from $\tilde{\tau}_c^*$. The algorithm retains this
path only if
$\operatorname{States}(\tilde{\tau}_{c,\mathrm{suf}}^*)
\cap Q_{c,F}\neq\emptyset$.
Therefore,
$\tau_{\mathrm{suf}}^*
=\operatorname{Proj}_{T}
(\tilde{\tau}_{c,\mathrm{suf}}^*)$
is a candidate accepting suffix cycle according to Definition
\ref{def: candidate cycle}.

Moreover, $\tau_{\mathrm{pre}}^*$ is reconstructed from $q_{c,1}$ by
Algorithm \ref{alg: reconstruct}. Hence, Lemma
\ref{lem: accepting prefix suffix} applies and gives
$\operatorname{trace}(\tau^*)\models\varphi$. Thus, condition (i) of
Problem \ref{problem: 1} is satisfied.

Next, because $q_{c,2}=(\pi_{\otimes},c_2)\in Q_c$ and
$Q_c\subseteq\Pi_{\otimes}\times[0,d]$, it follows that $c_2\leq d$.
By Definition \ref{def: cost automaton}, the cost accumulated along the
reconstructed prefix from an initial state to $q_{c,1}$ is $c_1$.
Similarly, the cost accumulated along the finite path from $q_{c,1}$
to $q_{c,2}$, which corresponds to one traversal of
$\tau_{\mathrm{suf}}^*$, is $c_2-c_1$. Therefore,
$J(\tau^*)=c_1+(c_2-c_1)=c_2\leq d$.
Thus, condition (ii) of Problem \ref{problem: 1} is satisfied.

Finally, for each candidate accepting suffix cycle, Algorithm
\ref{alg: main} computes its long-run visit proportion according to
Definition \ref{def: long run visit}. Since the long-run visit
proportion of a prefix--suffix path is determined by its periodically
repeated suffix, the value associated with the returned path satisfies
\[
P_v(\tau^*)
=
\frac{
\operatorname{num}
\bigl(v,\operatorname{trace}(\tau_{\mathrm{suf}}^*)\bigr)|v|
}{
\left|\operatorname{trace}(\tau_{\mathrm{suf}}^*)\right|
}.
\]
Furthermore, whenever $\mathit{best\_plan}$ is updated, the algorithm
sets
$\mathit{best\_diff}
=|P_v(\tau^*)-P_v^{\mathrm{des}}|$ for the corresponding selected plan.
Since the algorithm returns $\tau^*$ only if
$\mathit{best\_diff}\leq\delta$, we obtain
$|P_v(\tau^*)-P_v^{\mathrm{des}}|\leq\delta$.
Thus, condition (iii) of Problem \ref{problem: 1} is satisfied.

Consequently, $\tau^*$ satisfies all three conditions of Problem
\ref{problem: 1} and is therefore a feasible solution.
\end{proof}

Theorem~\ref{thm: solution correctness} guarantees that
any path returned by Algorithm~\ref{alg: main}
is a feasible solution to Problem~\ref{problem: 1}.
Intuitively, the three requirements of Problem~\ref{problem: 1} are ensured as follows. 
Lemma~\ref{lem: accepting prefix suffix}
ensures that the returned prefix--suffix path satisfies
the LTL specification. The cost-automaton construction
guarantees that its overall cost does not exceed $d$,
while the algorithm's return condition ensures that
its long-run visit proportion differs from
$P_v^{\mathrm{des}}$ by at most $\delta$.
We next investigate the optimality of the solution after Steps 1--3 in Algorithm \ref{alg: main} with respect to the desired long-run visit proportion.
Before studying the optimality of the solution, we
first define its feasible domain as follows.
\begin{definition}
\label{def: product feasible paths}
Consider the product system $T_\otimes$ in
Definition~\ref{def: product system}, with initial-state
set $\Pi_{\otimes,0}$ and accepting-state set
$\Pi_{\otimes,F}$.
For the cost bound $d\in \mathbb{R}_{>0}$ in Problem~\ref{problem: 1},
define the set of cost-bounded accepting WTS paths
induced by $T_\otimes$ as
\begin{equation}
\label{eq: product feasible WTS paths}
\begin{aligned}
\mathcal{F}_d^\otimes
:= \bigl\{
&\tau=\operatorname{Proj}_{T}^{\otimes}(\tau_\otimes)
\;\bigm|\;
\tau_\otimes
=\tau_{\otimes,\mathrm{pre}}
\odot(\tau_{\otimes,\mathrm{suf}})^\omega
\\
&\in\operatorname{Path}^{\omega}(T_\otimes),
\quad
\tau_{\otimes,\mathrm{pre}}(0)\in\Pi_{\otimes,0},
\\
&\tau_{\otimes,\mathrm{suf}}\text{ forms a cycle as in
Definition~\ref{def: pre-suf struct}},
\\
&\operatorname{States}(\tau_{\otimes,\mathrm{suf}})
\cap\Pi_{\otimes,F}\neq\emptyset,
~ J(\tau)\le d
\bigr\},
\end{aligned}
\end{equation}
where $J$ is defined in
Definition~\ref{def: pre-suf struct}.
\hfill \rule{1.5mm}{1.5mm}
\end{definition}





Intuitively, $\mathcal{F}_d^\otimes$ consisting of WTS prefix--suffix
paths induced by cost-bounded accepting paths in the product system. 
Next, the optimality analysis under the domain $\mathcal{F}_d^\otimes$ is discussed in the following corollary.
\begin{corollary}
\label{cor: product optimality}
Suppose that Steps 1--3 of Algorithm \ref{alg: main} yield
$\mathit{best\_plan}\neq\bot$. Let
$\bar{\tau}
=
\bar{\tau}_{\mathrm{pre}}
\odot
(\bar{\tau}_{\mathrm{suf}})^\omega$
be the prefix--suffix path induced by $\mathit{best\_plan}$ after implementing Steps 1--3. Then
$\bar{\tau}$ minimizes the deviation from the desired long-run visit
proportion among all cost-bounded accepting prefix--suffix paths
represented by the product system; that is,
\begin{equation}
\left|P_v(\bar{\tau})-P_v^{\mathrm{des}}\right|
=
\min_{\tau\in\mathcal{F}_d^{\otimes}}
\left|P_v(\tau)-P_v^{\mathrm{des}}\right|,
\label{eq: product optimality}
\end{equation}
where $\mathcal{F}_d^{\otimes}$ denotes the set of all prefix--suffix
WTS paths induced by reachable accepting paths in the product system
whose costs do not exceed $d$, as introduced in Definition \ref{def: product feasible paths}.
\hfill \rule{1.5mm}{1.5mm}
\end{corollary}

\begin{proof}
Suppose, for contradiction, that there exists
$\hat{\tau}\in\mathcal{F}_d^{\otimes}$ such that
\begin{equation}
\left|P_v(\hat{\tau})-P_v^{\mathrm{des}}\right|
<
\left|P_v(\bar{\tau})-P_v^{\mathrm{des}}\right|.
\label{eq: product optimality contradiction}
\end{equation}

Since $\hat{\tau}\in\mathcal{F}_d^{\otimes}$, its reachable accepting
product-system suffix cycle induces a finite path in the cost automaton
between two states
$\hat q_{c,1}=(\hat\pi_{\otimes},\hat c_1)$ and
$\hat q_{c,2}=(\hat\pi_{\otimes},\hat c_2)$ satisfying
$\hat c_1<\hat c_2\leq d$. Moreover, the corresponding finite
cost-automaton path contains at least one state in $Q_{c,F}$.

Steps 1--3 of Algorithm \ref{alg: main} enumerate every product state,
every admissible pair of accumulated costs, and every finite
cost-automaton path between the corresponding cost-augmented states.
Therefore, the candidate suffix associated with $\hat{\tau}$ is
examined by the algorithm, and its deviation
$|P_v(\hat{\tau})-P_v^{\mathrm{des}}|$ is compared with
$\mathit{best\_diff}$.

Since $\mathit{best\_plan}$ stores a candidate with the smallest
deviation among all enumerated candidates, the selected path
$\bar{\tau}$ must satisfy
\begin{equation}
\left|P_v(\bar{\tau})-P_v^{\mathrm{des}}\right|
\leq
\left|P_v(\hat{\tau})-P_v^{\mathrm{des}}\right|,
\end{equation}
which contradicts \eqref{eq: product optimality contradiction}.
Therefore, no such $\hat{\tau}$ exists, and
\eqref{eq: product optimality} holds.
\end{proof}
Corollary~\ref{cor: product optimality} establishes
the optimality of the path selected by Steps 1--3 of
Algorithm~\ref{alg: main} within $\mathcal{F}_d^{\otimes}$
with respect to the deviation from the desired long-run
visit proportion.
Intuitively, Steps 1--3 of Algorithm~\ref{alg: main}
systematically examine all finite cost-automaton paths
that induce candidate accepting suffix cycles of WTS, as introduced in Definition \ref{def: candidate cycle},
evaluate their associated long-run visit proportions,
and retain a plan with the smallest deviation from
$P_v^{\mathrm{des}}$.

\begin{table*}[t]
\centering
\caption{Experimental results under different desired long-run visit proportions.}
\label{tab:experimental results}
\begin{tabularx}{\textwidth}{
    >{\centering\arraybackslash}p{1.7cm}
    >{\raggedright\arraybackslash}X
    >{\centering\arraybackslash}p{1.7cm}
    >{\centering\arraybackslash}p{1.7cm}
    >{\centering\arraybackslash}p{1.7cm}
}
\toprule
Desired $P_v^{\mathrm{des}}$ &
The prefix--suffix path $\bar{\tau}
=
\bar{\tau}_{\mathrm{pre}}
\odot
(\bar{\tau}_{\mathrm{suf}})^\omega$
induced by $\mathit{best\_plan}$ after implementing Steps 1--3 in Algorithm \ref{alg: main} &
 $P_v(\bar{\tau})$ &
  Cost $J(\bar{\tau})$&
 Difference $\lvert P_v-P_v^{\mathrm{des}}\rvert$\\
\midrule
$0.200$ & $(q_0q_2q_1q_0q_1q_0q_1q_0q_1q_0q_2q_3q_0q_2q_1)^{\omega}$ & $0.200$&
19
& $0.000$
\\

$0.300$ & $(q_0q_1q_0q_1q_0q_2q_3q_0q_2q_1)^{\omega}$ & 
$0.300$ & 13
& $0.000$
\\

$0.500$ & $(q_0q_2q_3q_0q_2q_1q_0q_2q_3q_0q_2q_1)^{\omega}$ & $0.500$& 18
& $0.000$
\\

$0.700$ & $(q_0q_2q_3q_0q_2q_3q_0q_2q_1q_0q_1q_0q_2q_3q_0q_2q_3)^{\omega}$ & $0.706$ & 26
& $0.006$
\\

$0.900$ & $(q_0q_2q_3q_0q_2q_3q_0q_2q_3q_0q_2q_1q_0q_2q_3q_0q_2q_3)^{\omega}$ & $0.833$ & 29
& $0.067$
\\
$1.000$ & $(q_0q_2q_3q_0q_2q_3q_0q_2q_3q_0q_2q_1q_0q_2q_3q_0q_2q_3)^{\omega}$ & $0.833$ & 29
& $0.167$
\\
\bottomrule
\end{tabularx}
\end{table*}

\section{Experiment} \label{sec: experiment}
In this section, we implement the proposed approach on a Unitree B2 quadruped robot. The motion capabilities of the robot are abstracted by the same WTS $T$ introduced in Example \ref{eg: motivation}, as shown in Fig.\ref{fig: motivating example}. The associated experimental environment is shown in Fig. \ref{fig:env}, where the initial, upload, recharge, and gather regions are marked in dusty blue, pink, yellow and light blue, respectively.

Consider the same LTL specification $\varphi =\Box \Diamond \text{gather} ~\wedge ~\Box \Diamond \text{recharge} ~\wedge ~\Box \Diamond \text{upload}$ as described in Example \ref{eg: motivation}.
The AP sequence of interest introduced in Definition \ref{def: occ num} is chosen as $v =\{recharge\}, \emptyset, \{gather\} $, which corresponds to the WTS state sequence $q_{3}q_{0}q_{2}$, whose long-run visit proportion is regulated by specifying different desired values $P_v^{\mathrm{des}}$.
The cost $J(\tau)$, defined in Definition \ref{def: pre-suf struct}, represents the energy consumption associated with the prefix–suffix structure of $\tau$ and is constrained by $J(\tau)\leq 30$, where $30$ denotes the maximum allowable energy consumption.

Based on the WTS $T$ and the LTL specification $\varphi$, we apply the proposed approach according to Theorem \ref{thm: solution correctness} and Algorithm \ref{alg: main}. 
We conduct experiments with different desired long-run visit proportions $P_v^{\mathrm{des}}$ for the AP sequence of interest $v$, while fixing the allowable deviation from the desired proportion at $\delta=0.1$, as shown in Table I.

As shown in Table \ref{tab:experimental results}, the planned path is affected by the desired long-run visit proportion $P_v^{\mathrm{des}}$. As $P_v^{\mathrm{des}}$ increases from $0.2$ to $1.0$, the AP sequence of interest occupies an increasingly large proportion of the resulting prefix–suffix path, with the achieved proportion $P_v(\bar{\tau})$ increasing from $0.2$ to $0.833$. Nevertheless, $P_v(\bar{\tau})$ cannot reach $1$, because exclusively repeating the AP sequence of interest would prevent the resulting path from satisfying the LTL specification, which also requires the upload region to be visited infinitely often. Consequently, when $P_v^{\mathrm{des}}=1$, the difference $0.167$ exceeds the prescribed tolerance $\delta=0.1$. Despite this deviation, the obtained plan achieves the maximum attainable long-run visit proportion of the AP sequence of interest among all paths satisfying both the LTL specification and the cost constraint.

The experimental video can be found at \href{https://vimeo.com/1227934005}{https://vimeo.com/1227934005}.

\begin{figure}[htbp]
	\centering
	\includegraphics[width=0.35\textwidth]{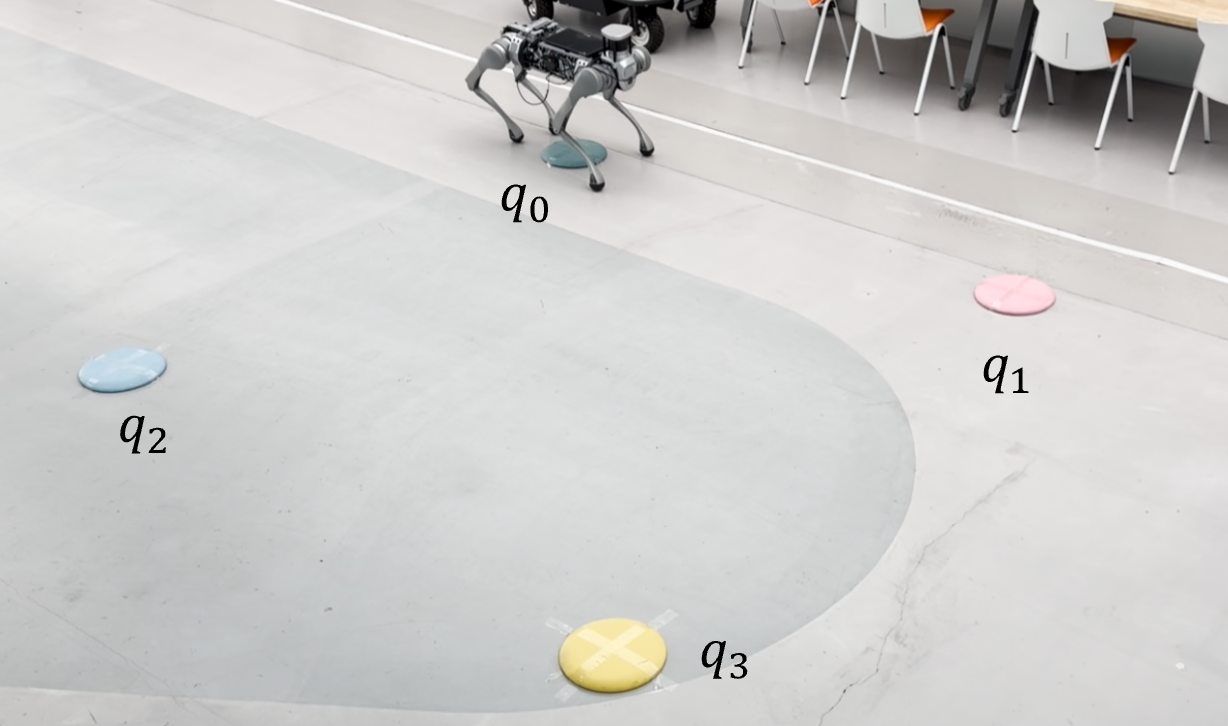}
    \vspace{-4pt}
	\caption{Experiment environment configuration.}
    \vspace{4pt}
	\label{fig:env}
\end{figure}

\section{Conclusion} \label{sec: conclude}
In this paper, we introduced a quantitative notion, termed the long-run visit proportion, for path planning under LTL specifications. Specifically, this notion characterizes the proportion of occurrences of an AP sequence of interest within the suffix cycle of an infinite prefix–suffix path. We further proposed a path-planning method that seeks a path whose long-run visit proportion is within a prescribed tolerance of a desired value, while satisfying the LTL specification and the cost constraint. Future work will extend the proposed notion and planning approach to stochastic and multi-agent systems.

\bibliographystyle{ieeetr}
\bibliography{reference}
\end{document}